\documentclass[11pt]{article}
\usepackage[english]{babel}
\usepackage{fullpage}
\usepackage[utf8]{inputenc}

\title{Fair Division with Social Impact}

\author{Michele Flammini$^{1,2}$, Gianluigi Greco$^2$, Giovanna Varricchio$^2$}
\date{%
	$^1$Gran Sasso Science Institute, L'Aquila, Italy\\%
	$^2$University of Calabria, Rende, Italy\\%
		michele.flammini@gssi.it, gianluigi.greco@unical.it, giovanna.varricchio@unical.it
	\\[2ex]%
	\today
}

\usepackage{multirow}
\usepackage{cellspace}
\usepackage{amsmath,amsthm,amsfonts,mathtools}
\usepackage{xcolor}
\usepackage{thm-restate}
\usepackage{enumitem}
\usepackage{cleveref}
\usepackage{tcolorbox}
\usepackage[linesnumbered,ruled,vlined]{algorithm2e}
\newtheorem{proposition}{Proposition}
\newtheorem{observation}{Observation}
\newtheorem{theorem}{Theorem}

\newtheorem{corollary}{Corollary}

\theoremstyle{definition}
\newtheorem{example}{Example}
\newtheorem{definition}{Definition}
\newtheorem{fact}{Fact}
\renewcommand{\epsilon}{\varepsilon}
\newcommand{\agents}{\mathcal{N}}
\newcommand{\goods}{\mathcal{G}}

\newcommand{\allocation}{\mathcal{A}}

\newcommand{\SW}{\mathsf{SW}}
\newcommand{\set}[1]{\{#1\}}

\newcommand{\modulus}[1]{\vert #1 \rvert }

\newcommand{\OPT}{\textsf{OPT}}
\newcommand{\opt}{\textsf{opt}}
\newcommand{\instance}{\mathcal{I}}

\newcommand{\classNP}{\textsf{NP}}

\newcommand{\classP}{\textsf{P}}

\newcommand{\EF}{\mathsf{EF}}
\newcommand{\EFX}{\mathsf{EFX}}
\newcommand{\EFk}{\mathsf{EFk}}
\newcommand{\PROP}{\mathsf{PROP}}
\newcommand{\PS}{\mathsf{PS}}

\newcommand{\SEF}{\mathsf{sEF}}

\newcommand{\maxUt}{\textsc{maxUt}}
\newcommand{\maxSW}{\textsc{maxSW}}
\newcommand{\maxNash}{\textsc{maxNash}}

\newcommand{\ut}{\textsc{Ut}}
\newcommand{\eg}{\textsc{Eg}}
\newcommand{\nash}{\textsc{Nash}}

\newcommand{\cycleElimination}{\textsc{CycleEliminiation}}
\newcommand{\topOrd}{\textsc{topOrd}}
\DeclareMathOperator{\argmax}{arg\,max}

\begin{document}
\maketitle

\begin{abstract}
	In this paper, we consider the problem of fair division of indivisible goods when the allocation of goods impacts society. Specifically, we introduce a second valuation function for each agent, determining the social impact of allocating a good to the agent. Such impact is considered desirable for the society -- the higher, the better. Our goal is to understand how to allocate goods fairly from the agents' perspective while maintaining society as happy as possible.
	To this end, we measure the impact on society using the utilitarian social welfare and provide both possibility and impossibility results.
	Our findings reveal that achieving good approximations, better than linear in the number of agents, is not possible while ensuring fairness to the agents. These impossibility results can be attributed to the fact that agents are completely unconscious of their social impact. 
	Consequently, we explore scenarios where agents are socially aware, by introducing related fairness notions, and demonstrate that an appropriate definition of fairness aligns with the goal of maximizing the social objective.
\end{abstract}


\section{Introduction}
Fair division is a fundamental research area at the intersection of economics and computer science. In the last decade, it turned out to be a particularly active field due to its applications to several real-world contexts.   It focuses on distributing resources (goods) among individuals (agents) in a fair manner according to their valuations. To achieve fairness, various criteria have been studied, including \emph{envy-freeness} ($\EF$)~\cite{foley1966resource}, where every agent weakly prefers her bundle over the bundle received by any other, and \emph{proportionality} ($\PROP$)~\cite{steinhaus1948problem}, where every agent values her bundle at least as much as her proportional share, that is, the value she attributes to the whole set of goods normalized by the number of agents.  Unfortunately, in indivisible settings, $\EF$ and $\PROP$ allocations do not always exist. Consequently, several relaxations have been considered. These include, but are not limited to, the {\em up to one} \cite{budish2011combinatorial}, {\em up to any} \cite{caragiannis2019unreasonable}, {\em epistemic} \cite{aziz2018knowledge,caragiannis2023new}, and {\em randomized} \cite{azizFreeman2023best} variants of $\EF$, $\PROP$, and related notions. %
The fair division model has also been extended to deal with more complex scenarios, e.g., endowing agents with entitlements~\cite{suksompong2024weighted, chakraborty2021weighted, aziz2023best, hoefer2024best}
or introducing additional constraints \cite{bilo2022almost,barman2023guaranteeing,bredereck2022envy}.

In most of this literature, the focus has been put on achieving fairness for the involved agents, possibly subject to efficiency, without considering any kind of externality. However, there are many scenarios in which allocating resources has an impact on society.
We can imagine, for example, there are some green strategies that the government must assign to companies. A company will profit from implementing the strategy (as companies more sensitive to green aspects are increasingly appealing to customers), while the implementation will positively impact society by reducing emissions. Companies judge the outcome based solely on the profit they and other companies can attain, whereas society aims to maximize the overall reduction of emissions.

With this paper, we expand the scope of fair division of indivisible goods incorporating the concept of social impact. Our goal is twofold: To ensure fairness among agents and guarantee a high social impact.

\paragraph{Our Contribution.}
\begin{table*}[tb]
	\centering
	\begin{tabular}{l c c c c c }
		Price of &$\EF1$ & $\EFX$ & $\EF2$ & $\PROP1$ & epistemic $\EF1$ 
		\\
		\hline 
		\\[-1.5\medskipamount] 
		\multirow{2}{*}{UB}
		& {$O(n^2)$ {\footnotesize (Thm.~\ref{thm:EF1n2})}}
		& \multirow{2}{*}{$n^\star$ {\footnotesize (Thm.~\ref{thm:approxIdentical})}}
		& \multirow{2}{*}{$n$ {\footnotesize (Thm.~\ref{thm:EF2approx})}}
		& \multirow{2}{*}{$n$ {\footnotesize (Thm.~\ref{thm:approxEpistemic})}}
		& \multirow{2}{*}{$n$  {\footnotesize (Thm.~\ref{thm:approxEpistemic})}}
		\\[0.1cm]
		& {$n^{\dagger}$ {\footnotesize (Thm.~\ref{thm:EF1approxOrd})}} &&&& \\ [0.3cm]
		{LB}
		& {$n$ {\footnotesize (Thm.~\ref{thm:LButEFk})}}
		& {$n$ {\footnotesize (Cor.~\ref{cor:LButEFX})}}
		& {$n-1$ {\footnotesize (Thm.~\ref{thm:LButEFk})}}
		& {$n$ {\footnotesize (Cor.~\ref{cor:LButPROP1})}}
		& {$n$  {\footnotesize (Cor.~\ref{cor:LButPROP1})}} 
		\\[0.1cm]
		\hline
	\end{tabular}
	\caption{Lower (LB) and upper (UB) bounds to the PoF w.r.t.\ the maximum utilitarian social impact. The results marked with $^\dagger$ and $^\star$ hold for ordered and identical valuations, respectively. All the LBs have been proven even under identical valuations.}
	\label{tab:PoF}
\end{table*}

We refer to the preliminaries for the formal definition of the considered fairness criteria.

In this work, we aim at determining the loss in efficiency (expressed by the utilitarian welfare w.r.t.\ the social impact of agents, i.e., the sum of the social impact of agents) while pursuing fairness. In particular, we estimate the {\em Price of Fairness} (PoF), that is, the worst-case ratio between the value of the maximum utilitarian welfare and the maximum utilitarian welfare achievable by a fair solution. 
Besides determining suitable bounds for several fairness notions, we also provide efficient algorithms for computing the corresponding approximately optimal and fair allocations.  
An overview of our results for $n$ agents is given in~\Cref{tab:PoF}.

For all the fairness notions we considered, an approximation better than linear in the number of agents is not possible. 
This can be attributed to the fact that agents are unconscious of their social impact. We therefore investigate which kind of consciousness on the agents' side allows the existence of fair and optimal outcomes. To this end, we introduce the notion of {\em socially aware envy-freeness} ($\SEF$) and establish that an allocation maximizing the utilitarian social impact and $\SEF$ {\em up to one good} always exists.
Due to space constraints, we refer to the full paper for the details.
\paragraph{Related Work.}
The literature on fair division is wide; we refer the interested reader to some recent surveys~\cite{amanatidis2023fair, suksompong2021constraints,suksompong2024weighted}. We next discuss the work directly related to our.

The goal of achieving fairness while maximizing the utilitarian welfare has been extensively studied when $s_i = v_i$.
Allocations having maximum agents' utilitarian social welfare
do not necessarily guarantee fairness, even in divisible settings. For this reason, the PoF with respect to such a welfare has been extensively studied \cite{nicosia2017price,barman2020optimal,bei2021price,li2024complete}. 
In this stream of research, several fairness notions as well as restrictions on the valuations have been exploited. 
Furthermore, the problem of determining whether there exists a fair allocation that is also socially optimal has been considered in \cite{aziz2020polynomial, bu2022complexity}. In a nutshell, they show that for all the considered criteria, both determining the existence of a maximum utilitarian allocation which is also fair and computing an allocation that maximizes the utilitarian welfare among fair solutions are strongly NP-hard problems.

The idea of employing two distinct valuation functions for each agent has been examined in prior research~\cite{bu2023fair} albeit from a different viewpoint. In this approach, one valuation function captures the agent's preferences while the other represents the allocator's interpretation of the agent's preferences -- what the allocator thinks a bundle of goods should value for a certain agent. The main objective in~\cite{bu2023fair} was to achieve fairness in both the agents' and the allocator's perspectives. 
In our model, society does not care if the allocation is fair according to the social impact functions; rather, society would like the allocation to maximize its own interests. In~\cite{bu2023fair}, the authors briefly discuss the problem of maximizing the allocator's utilitarian welfare by providing a simple $m$-approximation algorithm, where $m$ is the number of goods. With our contribution, we improve this result achieving an approximation polynomial in the number of agents which gives a significant improvement as the number of goods might be extremely larger than, even exponential in, the number of agents.


\section{Preliminaries}
In this section, we provide the basics of fair division as well as a formal description of the framework with social impact.
Hereafter, we denote by $[t]$, for $t\in \mathbb{N}$, the set $\set{1, \dots, t}$.

In fair division of indivisible goods, we are given a set of $n$ {\em agents} $\agents$ and a set of $m$ {\em goods} $\goods$, also called {\em items}, that have to be allocated to agents in a fair manner.
An allocation $\allocation$ consists of disjoint bundles $(A_1, \dots, A_n)$, where each $A_i$ is the set of goods, the bundle, allocated to agent $i$. Unless stated otherwise, we assume allocations to be complete, meaning that all goods have been allocated, i.e., $\cup_{i\in \agents}A_i =\goods$. If an allocation is not complete, we call it \emph{partial}.
The fairness of an allocation is determined according to the agents' preferences; for each agent $i\in \agents$ there exists a \emph{valuation function} $v_i:2^\goods \rightarrow \mathbb{R}_{\geq 0}$ mapping bundles of goods to non-negative real values. A fair division instance is defined by the triple $\instance = (\agents, \goods, \set{v_i}_{i\in \agents})$.
In this paper, we focus on {\em additive valuations}, i.e., for each bundle $A \subseteq \goods $, $v_i(A)= \sum_{g\in A}v_i(\set{g})$. For the sake of simplicity, we write $v_i(g)$ to denote $v_i(\set{g})$. Moreover, we assume $v_i(\emptyset)=0$.

The most relevant and intuitive notion of fairness is the one of {\em envy-freeness} ($\EF$), which postulates that, in a given allocation $\allocation$, no agent envies the bundle of any other agent, i.e.,  for each $i,j\in \agents$, $v_i(A_i)\geq v_i(A_j)$. Unfortunately, such an allocation might not exist even in simple cases: Consider an instance with two agents and one good, no matter who will receive the good, the other will be envious.
Thus, several relaxations have been proposed. Among the others: 

An allocation $\allocation$ is said to be {\em envy-free up to one good} ($\EF1$) if for each pair $i,j$ of agents, either $A_j=\emptyset$, or $\exists g\in A_j$ such that $v_i(A_i)\geq v_i(A_j\setminus\set{g})$. 

An allocation $\allocation$ is said to be {\em envy-free up to any good} ($\EFX$) if for each pair of agents $i,j$ either $A_j=\emptyset$ or, for each $g\in A_j$, $v_i(A_i)\geq v_i(A_j\setminus\set{g})$.

While $\EF1$ allocations are known to always exist, even for monotone valuations, the existence of $\EFX$ allocations remains a fundamental open question for additive valuations.
The $\EF1$ definition has been further extended allowing the ipotetical removal of at most $k$ goods from the envied bundle, in this case, we talk about {\em envy-free up to $k$ goods }($\EFk$).

Besides comparison-based fairness notions like $\EF$, shared-based ones have also been introduced. Among these, {\em proportionality} ($\PROP$) guarantees to each agent at least her proportional share; formally, for each $i\in \agents$, $v_i(A_i)\geq \PS_i  \coloneqq  \frac{1}{n}\cdot v_i(\goods)$. Achieving proportionality,  as for $\EF$, might not be possible. Hence, {\em proportional up to one good}  ($\PROP1$), the analogous of $\EF1$ for proportionality, has been proposed. 
$\allocation$ is said to be $\PROP1$ if for each $i$ there exists $g\not\in A_i$ such that  $v_i(A_i \cup \set{g})\geq \PS_i$. For additive valuations, is known $\EF \Rightarrow \PROP$ and $\EF1 \Rightarrow \PROP1$. 

We finally define the epistemic variant of $\EF1$; $\allocation$ is said to be {\em epistemic} $\EF1$ if, for each $i\in\agents$, there exists an allocation $\allocation^*$ such that $A_i^* = A_i$ and $\allocation^*$ is $\EF1$ in $i$'s perspective. Such an allocation $\allocation^*$ is also called the epistemic $\EF1$ certificate (for agent $i$). In other words, for each agent, there exists an allocation where the bundle $A_i$ makes her $\EF1$ satisfied. Because of the existence of this certificate for each agent $i$, being $\EF1\Rightarrow\PROP1$, we know that $A_i$ makes $i$ $\PROP1$ (as share-based notions solely depend on the received bundle); hence, epistemic $\EF1\Rightarrow\PROP1$.

In the sequel, we focus on additive valuations. An interesting subclass is the one of {\em ordered valuations}, where it is further assumed that there exists a fixed ordering of the goods $g^1, \dots, g^m$ such that, for each agent $i$ and  $h,k \in [m]$ with $h<k$, it holds $v_i(g^h)\geq v_i(g^k)$.
A special case of ordered instances is the one of {\em identical valuations}, where there exists an additive function $v$ s.t.\ $v_i=v$, for each $i\in\agents$.

To understand our algorithms, we first overview standard approaches to compute $\EF1$ allocations.
\paragraph{The Envy-Cycle Elimination.}
The envy-cycle elimination algorithm, introduced by~\cite{lipton2004approximately}, starts with an empty allocation $\allocation$.
At each round, one available good $g$ is allocated to an agent $i$ who is not envied by any other agent.
After allocating the good $g$, if there exists an agent $j$ envying agent $i$, the $\EF1$ is satisfied as agent $j$ was not envious before the allocation of $g$. 
To ensure the existence of such an agent $i$, there is
a subroutine, the \cycleElimination, which takes as input the current partial allocation $\allocation$ and the corresponding {\em envy-graph} $G$ where nodes are the agents and there is a directed edge from $i$ to $j$ if $i$ envies $j$ in $\allocation$. The \cycleElimination\ finds a cycle in $G$
and trades bundles along the cycle, that is, for each directed edge ${i,j}$ in the cycle, agent $i$ receives the bundle $A_j$,
and it is run until there is at least one agent who is not envied by anyone else.
\cycleElimination\ preserves the $\EF1$condition among agents, even for monotone valuations.
\paragraph{Sequential Allocation Algorithms.}
To describe sequential allocation algorithms, we first introduce the concept of {\em picking sequence}. A picking sequence $\sigma= (\sigma(1), \dots, \sigma(m))$ is a sequence of $m$ entries where $\sigma(k)\in \agents$, for each $k\in [m]$. A {\em prefix} $\sigma^k$, for $k\in [m]$, is the subsequence $(\sigma(1), \dots, \sigma(k))$. A sequence $\sigma$ is said to be {\em recursively balanced} (RB) if for any prefix $\sigma^k$ and any pair of agents $i,j$, denoted by $p_i$ and $p_j$ the number of occurrences of $i$ and $j$ in $\sigma^k$, respectively, it holds $\modulus{p_i-p_j}\leq 1$.

A sequential allocation algorithm takes as input a fair division instance and a picking sequence of length $m$, where $m$ is the number of goods to be allocated, and proceeds in $m$ steps. It starts with an empty allocation. At step $k$, agent $\sigma(k)$ choose the most preferred available good which is put in her bundle and removed from the set of available ones.

\begin{fact}[From~\cite{azizConstrainedRR}]\label{fact:RBisEF1}
	For additive valuations, if $\sigma$ is RB, then, the sequential algorithm returns an $\EF1$ allocation.
\end{fact}

Given an ordering of the agents $1, \dots, n$, a well-known RB picking sequence is $\sigma=(1, \dots, n,$ $ 1, \dots, n, 1, \dots)$ whose corresponding sequential algorithm is the famous \emph{round-robin} (RR) algorithm.

\paragraph{Social Impact.}
In this paper, we assume there exists an additive {\em social impact} function $s_i:2^\goods \rightarrow \mathbb{R}_{\geq0}$, which represent the happiness of the society for allocating a bundle of goods to an agent $i\in\agents$. We assume, the higher the social impact, the better is for society.
We pay particular attention to the problem of maximizing the {\em utilitarian} social welfare, given by the sum of the agents' social impact, that is, $\SW(\allocation)=\sum_{i} s_i(A_i)$. 
We denote by \maxUt\ the problem of maximizing the utilitarian welfare, by $\OPT$ any optimal allocation for a given instance of \maxUt\, and by $\opt=\SW(\OPT)$ its social welfare.


\section{Impossibility Results}\label{sec:impossibilityMaxUt}

Let us start by observing that it is possible to efficiently compute an optimal solution of \maxUt\ if the $\EF1$ condition is not a constraint. It is sufficient to allocate each good $g\in \goods$ to an agent $i^*$ who has the highest social impact for it, that is, to $i^*\in \argmax_{i\in\agents} s_i(g)$. 
Nonetheless, an approximation better than $n$ is not possible for \maxUt\ while requiring $\EF1$. This is a direct consequence of the following result.

\begin{theorem}\label{thm:LButEFk}
	An approximation better than $n-k+1$ to $\opt$ is not possible when requiring $\EF k$, even for identical agents.
\end{theorem}
\begin{proof}
	Consider an instance with $n$ agents having identical valuation functions and $m=h\cdot n +k -1$ goods, where $h=n - k$. All agents value $1$ every good, agent $1$ has a social impact of $1$ for each good, and all remaining ones social impact $0$. In such an instance $\opt=m$ and it is obtained by assigning all goods to agent $1$. 
	
	Concerning $\EF k$ allocations, no agent can receive more than $h + k -1$ goods without violating the $\EF k$ condition, otherwise, there exist two bundles differing by more than $k$ goods -- a contradiction to $\EF k$. 
	
	In conclusion, agent $1$ cannot receive more than $h + k -1$ goods in an $\EF k$ allocation, and hence, no approximation better than
	$\frac{m}{h+k-1}= \frac{h\cdot n + k -1}{h+k-1}$
	is possible, where the last equality holds as $h=n - k$. 
\end{proof}
Since in the just described instance all agents value $1$ all goods, $\EF1$ is equivalent to $\EFX$ showing the following.

\begin{corollary}\label{cor:LButEFX}
	An approximation better than $n$ to \maxUt\ is not possible when requiring $\EFX$.
\end{corollary}

Moreover, in the instance described in the proof of \Cref{thm:LButEFk}, if $k=1$, the only way to get a $\PROP1$ or epistemic $\EF1$ allocation is to allocate to each agent the same number of goods -- in fact, for such instance, $\PROP1$ is equivalent to $\PROP$ and $\PROP1$ is a necessary condition for epistemic $\EF1$. These observations together imply the following.

\begin{corollary}\label{cor:LButPROP1}
	An approximation better than $n$ to \maxUt\ is not possible when requiring $\PROP1$ or epistemic $\EF1$.
\end{corollary}


\section{Approximating \maxUt\ subject to fairness}
Given the above impossiblity results, we now investigate what is the best achievable approximation to $\opt$ when requiring fairness. We first focus on ordered and identical valuations, and then turn the attention to the broader class of additive valuations.

\subsection{A general approx. algorithm for ordered valuations}
In this section, we present a general approximation algorithm providing $\EF1$ allocations in the case of ordered instances, proving the following theorem. 

\begin{theorem}\label{thm:EF1approxOrd}
	For ordered valuations, there exists a poly-time algorithm computing an $\EF1$ allocation that is an $n$-approximation of $\opt$.  
\end{theorem}

Assume w.l.o.g.\ $m=q\cdot n$, otherwise, we can add dummy goods of value and social impact  $0$ for all agents. The dummy goods will be ranked last from all agents maintaining the ordered valuations property.
Let $g^1, g^2,\cdots, g^m$ the common ordering of the goods according to the ordered valuations assumption. We refer to $g^k$ as the $k$-th good. 
Every sequential allocation algorithm we employ will assign $g^k$ to the $k$-th agent in the picking sequence order. For our purposes, we also partition $\goods$, according to their ordering, into blocks of cardinality $n$. In particular, 
we denote by $\goods^h = \set{g^{(h-1)n +1}, \cdots, g^{hn}}$, for $h\in[q]$. In other words, block $\goods^1$ contains the $n$ most preferred goods; block $\goods^2$ contains the most $n$ preferred goods among $\goods \setminus \goods^1$, and so forth.

\begin{restatable}{lemma}{identicalRB}\label{lemma:identicalRB}
	Given an ordered fair division instance, if $\allocation$ is an allocation in which each agent receives a bundle containing exactly one good in each $\goods^h$, $\forall h\in[q]$, then, $\allocation$ is $\EF1$.
\end{restatable}
This lemma can be proven noticing that an allocation satisfying such a condition is attainable by running a sequential algorithm with an appropriate RB picking sequence.

We next build our approximation algorithm, Algorithm~\ref{algo:EF1transfOrd}, around the conditions of \Cref{lemma:identicalRB}. 
The algorithm will take as input an
allocation $\allocation$ and transform it into an $\EF1$ allocation $\allocation'$ while maintaining suitable properties w.r.t.\ the social impact of $\allocation$. Specifically, Algorithm~\ref{algo:EF1transfOrd} proceeds into $q$ rounds and 
at round $h\in [q]$ assigns the goods in $\goods^h$, one for each agent. Firstly, for each agent $i$ with $\goods^h \cap A_i \neq \emptyset$, it puts in $A_i'$ the best good according to the social impact of $i$ among the ones in $\goods^h \cap A_i$, that is, a good in $\arg\max_{g\in \goods^h \cap A_i} s_i(g)$. 
The remaining agents with $\goods^h \cap A_i= \emptyset$ will receive one of the remaining goods in $\goods^h$. 
A more formal description is presented in Algorithm~\ref{algo:EF1transfOrd}.
Notice, the algorithm solely depends on the goods ordering and not on the agents' valuations.

\begin{algorithm}[htb!]
	\SetNoFillComment
	\DontPrintSemicolon
	\KwIn{An allocation $\allocation$ for a fair division instance with ordered valuations and with $m= q\cdot n $}
	\KwOut{An $\EF1$ allocation $\allocation'$}
	\For{$k=1, \dots, q$}{ 
		$X\gets \agents$ and $Y\gets \goods^k$\\
		\For{$i=1, \dots, n$}{
			\If{$\goods^k \cap A_i\neq \emptyset$}{
				$g \gets g\in \arg\max_{g\in \goods^k \cap A_i}s_i(g)$\\
				$X= X \setminus \set{i}$ and 
				$Y= Y \setminus \set{g}$\\
				$A'_i= A'_i \cup \set{g}$\\
			}
		}
		Each good in $Y$ is given to a distinct agent in $X$\\
	}
	\KwRet{$\allocation'$}\\
	\caption{Transformation into an $\EF1$ allocation\label{algo:EF1transfOrd}}
\end{algorithm}

\begin{restatable}{proposition}{genApproxAlgIdentical}\label{prop:genApproxAlgIdentical}
	If $\allocation$ is the input of Algorithm~\ref{algo:EF1transfOrd} and $\allocation'$ is the output, then, $\allocation'$ is $\EF1$. Moreover, $n \cdot s_i(A_i')\geq s_i(A_i)$.
\end{restatable}

\begin{proof}[Sketch]
	Algorithm~\ref{algo:EF1transfOrd}, by design, computes an allocation satisfying the hypothesis of \Cref{lemma:identicalRB}, and thus, $\allocation'$ is $\EF1$.
	
	About the approximation, at each round, agent $i$ receives the best, according to $s_i$, remaining goods of $A_i$ while the other agents receive at most one good from $A_i$. Thus, the social impact of the good agent $i$ receives is, in the worst case, an $n$-approximation to the goods of $A_i$ allocated in this round. Summing up for all rounds the thesis follows. 
\end{proof}

\Cref{thm:EF1approxOrd} follows by \Cref{prop:genApproxAlgIdentical} as the utilitarian welfare is the sum of the agents' utilities. \Cref{prop:genApproxAlgIdentical} has broader implications on the egalitarian, Nash welfare, and in general on $p$-means, where an $n$-approximation can be guaranteed as well.


\subsection{An $n$-approx.\ subject to $\EFX$ for identical agents}
In this section, we further assume that agents have identical valuations and  strengthen \Cref{thm:EF1approxOrd} as follows:

\begin{restatable}{theorem}{approxIdentical}\label{thm:approxIdentical}
	For identical valuations, there exists a poly-time algorithm computing an $\EFX$ allocation that is an $n$-approximation of $\opt$. 
\end{restatable}
\begin{proof}[Sketch]
	We recall that, for identical valuations, an $\EFX$ allocation can be computed in polynomial time \cite{barman2018greedy}.
	Let $\allocation$ be such an allocation. Since agents have identical valuations, any permutation of bundles of $\allocation$ remains $\EFX$. We can therefore compute a maximum matching between agents and bundles of $\allocation$, where the weight of an agent-bundle edge is the social impact the agent would have by receiving that bundle. 
	Let us assume w.l.o.g.\ agent $i$ is matched to bundle $A_i$ in the maximum matching. 
	Clearly, any other matching cannot provide a better social impact. 
	Consider now, $n$ distinct matchings such that, for each $i$ and $j$, $i$ gets $A_j$ in exactly one of these matchings. Such matchings always exist as they can be obtained by shifting bundles among agents.
	Therefore, $n \cdot \sum_{i=1}^n s_i(A_i) \geq  \sum_{i=1}^n \sum_{j=1}^n s_i(A_j)\geq \opt$, where the last inequality holds as $s_i(\goods)= \sum_{j=1}^n s_i(A_j)$.
	
	In conclusion, the allocation corresponding to such maximum matching guarantees an $n$-approximation to $\opt$.
\end{proof}

\subsection{An $O(n^2)$-approximation subject to $\EF1$}\label{sec:apprximationMaxUt}
In this section, we tackle the problem of finding an $\EF1$ allocation with good approximation to $\opt$ for additive valuations.
The paper~\cite{bu2023fair} provides an easy $m$-approximation to $\opt$ subject to $\EF1$. Their algorithm, that we call \textsc{SimpleEF1Approx}, proceeds as follows: Finds a good $g^*$ of maximum social impact, that is, $\max_i s_i(g^*)\geq \max_{i,g} s_i(g)$.
Then, it assigns $g^*$ to the agent realizing the maximum social impact for it, say $i^*$. The remaining goods are allocated in an RR fashion w.r.t.\ an ordering where $i^*$ is the last picking. Clearly, optimally allocating only one good cannot guarantee an approximation to $\opt$ better than $m$.
Considering that the number of goods might be significantly larger $n$, we now improve the approximation as follows.

\begin{theorem}\label{thm:EF1n2}
	There exists a poly-time algorithm computing an $\EF1$ allocation providing a $O(n^2)$-approximation to $\opt$.
\end{theorem}

To this aim, our algorithm will take into account the real optimum to compute the $\EF1$ solution.
Let $\OPT$ be the optimum for \maxUt\ that will be used by our algorithm, and $\OPT_i$ be the bundle of $i$ in $\OPT$. Recall, $\opt$ is the optimal value $\SW(\OPT)$. 
Let $m_i$ be the number of goods in $\OPT_i$. We sort the social value of such goods from the highest to the lowest, namely, $\opt^1_i\geq \opt^2_i\geq \cdots \geq \opt^{m_i}_i$, and denote the good of value $\opt^k_i$ by $o^k_i$. 
Clearly, $\opt=\sum_{i\in \agents}\sum_{k=1}^{m_i}\opt^k_i$. 
In what follows, given any $i$ such that $m_i < n$, we add dummy objects $o^k_i$ valued $0$ by all the agents and with $\opt^k_i=0$,\ for every $m_i+1 \leq k \leq n$; removing such goods at the end will not affect $\EF1$ nor the approximation. Moreover, whenever we talk about the social value of a good, we mean the social value it has in $\OPT$.

Our algorithm is based on a case distinction depending on how the critical mass of $\opt$ is distributed among the bundles $\OPT_1, \cdots, \OPT_n$. Let us denote by $\Delta_1=\set{o^k_i \, \vert \, i\in\agents, k\in [n]}$, the $n$ socially best goods in $\OPT_i$, for each $i$. Note that $\Delta_1$ contains exactly $n^2$ goods. The social value of $\Delta_1$ is given by $\delta_1=\sum_{i=1}^n \sum_{k=1}^n \opt^k_i$. Furthermore, we indicate by $\Delta_2= \goods \setminus \Delta_1$ and denote by $\delta_2$ its social value. Hence, $\opt =\delta_1 +\delta_2$. 
We distinguish the following two cases:
\begin{enumerate}[leftmargin=1.5cm]
	\item[\bf Case 1)] $\delta_1 \geq \delta_2$;
	\item[\bf Case 2)] $\delta_2 < \delta_1$.
\end{enumerate}

\paragraph{An $2n^2$-approximation for Case~1.}
This case is the easiest; it suffices to use \textsc{SimpleEF1Approx} to get a $2n^2$-approximation of $\opt$. In fact, we can allocate the good $g^*$, of highest social impact, to the agent who owns it in $\OPT$ -- this means we are correctly allocating the good of highest social impact in $\OPT$. 
The social impact of such good provides an upper bound to every good in $\Delta_1$; since  $\modulus{\Delta_1} = n^2$, the outcome of \textsc{SimpleEF1Approx} is trivially an $n^2$-approximation of $\delta_1$. 
Being $\delta_1 \geq \frac{\delta_1+\delta_2}{2}= \frac{\opt}{2}$, this ensures an $2n^2$-approximation.

Despite the easy approach and analysis, Case~1 constitutes the bottleneck for the approximation of $\opt$.
In what follows, we focus on Case~2 which will require a more interesting algorithm providing a $2n$-approximation of $\opt$.

\paragraph{An $2 n$-approximation for Case~2.}
In this case, $\exists i\in\agents$ such that $\modulus{\OPT_i}>n$, otherwise, $\Delta_2=\emptyset$ and  $ \delta_2=0\leq \delta_1$.
Hence, we partition the goods of each $\OPT_i$ into groups of $n$ goods as long as it is possible. More precisely, if $m_i= k_i\cdot n + r_i$, with $r_i<n$, we create a group $B_i^k$ containing the goods $o_i^{(k-1)n +1}, \dots, o_i^{kn}$, for $k\in [k_i]$. The remaining goods, if any, are stored in $B_i^{k_i+1}$, otherwise, we let $B_i^{k_i+1}=\emptyset$. In the social impact perspective, any good in $B_i^k$ values at least as much as any good in $B_i^{k+1}$ when assigned to agent $i$. This observation is the key ingredient of the following.

\begin{restatable}{lemma}{Casetwosufficient}\label{lemma:Case2sufficient}
	In Case~2, any allocation that assigns to each agent $i\in \agents$ one good in $B_i^k$, for each $k=1, \cdots, k_i$, guarantees an $2n$-approximation to $\opt$.
\end{restatable}
\Cref{lemma:Case2sufficient} determines a sufficient condition for an allocation to be a $2n$-approximation to $\opt$. Algorithm~\ref{algo:Case2} will provide an $\EF1$ solution satisfying such a condition. This algorithm uses the same approach of \cite{bu2023fair} of partitioning $\goods$ in packages of $n$ goods. Their purpose was to achieve $\EF1$ both in the allocator and the agents perspectives, under the assumption that for all $i$ $s_i=s$ for some allocator valuation $s$. With our approach, we use a different partition of the goods into packages and our results hold for arbitrary $s_i$ and concerns the approximation facor of the resulting $\EF1$ allocation.

\begin{algorithm}[htb!]
	\SetNoFillComment
	\DontPrintSemicolon
	\KwIn{$\OPT$, $\goods$, $\agents$, $\set{v_i}_{i\in \agents}$, $\set{s_i}_{i\in \agents}$}
	\KwOut{An allocation $\allocation= (A_1, \cdots, A_n)$}
	\tcc{Phase 1: Preprocessing}
	Partition $\OPT_i$ in $B_i^1, \dots, B_i^{k_i+1}$ for each $i\in \agents$\\
	$\allocation \gets (\emptyset, \cdots, \emptyset)$\\
	\tcc{Phase 2: Compute a partial allocation satisfying \Cref{lemma:Case2sufficient}}
	Build the envy-graph $G$ corresponding to $\allocation$\\
	\For{$i\in\agents$ and $k=1, \cdots, k_i$}{ 
		$\sigma \gets \topOrd(G)$\\
		\For{$j=1, \dots , n$}{ \label{line:forLoopRound}
			$g\gets g\in \arg\max_{g\in B_i^k} v_{\sigma(j)}(g)$ \\
			$A_{\sigma(j)} \gets A_{\sigma(j)} \cup \set{g}$ and 
			$B_i^k \gets B_i^k  \setminus \set{g}$\\
		}
		Build the envy-graph $G$ corresponding to $\allocation$\\
		\While{$\topOrd(G) = False$}{
			$\allocation \gets \cycleElimination(\allocation, G)$ \\
			Build the envy-graph $G$ corresponding to $\allocation$
		}
	}
	\tcc{Phase~3: Allocate remaining goods w/o violating~\Cref{lemma:Case2sufficient}}
	$\mathcal{B}\gets \cup_{i\in\agents} B_i^{k_i+1}$\\
	Use envy-cycle elimination to allocate the goods $\mathcal{B}$\\
	\KwRet{$\allocation$}\\
	\caption{$O(n)$-approx. for \maxUt\ in Case~2\label{algo:Case2}}
\end{algorithm}

Informally, Algorithm~\ref{algo:Case2} takes an optimum allocation $\OPT$ of $\maxUt$ and in the {\em preprocessing} ({\em Phase 1}) partitions the bundle $\OPT_i$ into $B_i^1, \dots, B_i^{k_i+1}$ as we previously described. 
In {\em Phase 2}, the algorithm proceeds into rounds and ensures that, at the beginning of each round, the envy-graph corresponding to the current allocation is acyclic. We can therefore compute a topological order of the agents -- an agent in this ordering may envy only agents coming afterward -- which will be used in the current round to assign goods. 
In particular, the algorithm allocates in an RR fashion the goods of a certain $B_i^k$ with respect to the just computed topological order, for $k\leq k_i$. 
So no two goods in $B_i^k$ are assigned to the same agent. At the end of the round, the $\EF1$ condition is finally satisfied (interestingly, this might not be true in the middle of a round) and all the existing cycles in the envy-graph are eliminated. 
Finally, in {Phase 3}, the remaining goods $\mathcal{B}$ are allocated via the standard envy-cycle elimination algorithm.

The algorithm makes use of the $\topOrd(G)$ procedure which determines whether a directed graph $G$ admits a topological order, and if so, it outputs one.

In order to prove the correctness of the algorithm, we first observe that the output of Algorithm \ref{algo:Case2} satisfies \Cref{lemma:Case2sufficient}. In a nutshell, since it assigns to each bundle exactly one good from $B_i^k$, for every $i\in\agents$ and $k\in [k_i]$, no matter how bundles are subsequently shuffled, this property remains true. Therefore, a linear approximation immediately follows. 
In the remaining, we show the $\EF1$ condition. 
\begin{proposition}\label{prop:EF1Algo}
	{Phase 2} of Algorithm \ref{algo:Case2} is well-defined and produces an $\EF1$ allocation.
\end{proposition}
\begin{proof}
	We show at the end of each round the current allocation is $\EF1$.
	At the first round, the agents will have one good each and the allocation is $\EF1$.
	
	Assume we are in a generic round of {\em Phase 2} and the partial allocation $\allocation$ at the beginning of the round is $\EF1$. We have obtained $\allocation$ after repeatedly running, in the previous round, the procedure $\cycleElimination$ as long as the envy-graph is not acyclic. Hence, $\topOrd$ will now correctly output a topological order $\sigma$ of the agents.
	
	We first notice that the {\bf for} loop starting at line \ref{line:forLoopRound} assigns one good $x_i$ to each agent $i$ in a RR fashion with respect to the permutation $\sigma$. Hence, for each agent $\sigma(j)$, $v_{\sigma(j)}(x_{\sigma(j)})\geq v_{\sigma(j)}(x_{\sigma(j')})$ for each $j'>j$. Moreover, being $\sigma$ a topological ordering of the agents for the partial allocation $\allocation$, we have $v_{\sigma(j)}(A_{\sigma(j)})\geq v_{\sigma(j)}(A_{\sigma(j')})$ for each $j'<j$.
	At the end of the {\bf for} loop, agent $\sigma(j)$ owns the bundle $A'_{\sigma(j)} = A_{\sigma(j)} \cup \set{x_{\sigma(j)}}$. We next show that this bundle makes $\sigma(j)$ $\EF1$ at the end of the for loop. 
	
	Consider $j'<j$; being  $v_{\sigma(j)}(A'_{\sigma(j)} )\geq v_{\sigma(j)}(A_{\sigma(j)})$ and
	\begin{align*}
		v_{\sigma(j)}(A_{\sigma(j)})
		&\geq v_{\sigma(j)}(A_{\sigma(j')}) = v_{\sigma(j)}(A'_{\sigma(j')}\setminus \set{x_{\sigma(j')}}),
	\end{align*}
	agent $\sigma(j)$ turns out to be $\EF1$ w.r.t.\ $\sigma(j')$, as it suffices to remove the last inserted good from $A'_{\sigma(j')}$ to make $\sigma(j)$ $\EF$. 
	
	Consider now $j'>j$; since $\allocation$ is $\EF1$, $\exists g\in A_{\sigma(j')}$ s.t.
	\begin{align*}
		v_{\sigma(j)}(A_{\sigma(j)})\geq v_{\sigma(j)}(A_{\sigma(j')}\setminus \set{g});
	\end{align*}
	on the other hand $v_{\sigma(j)}(x_{\sigma(j)})\geq v_{\sigma(j)}(x_{\sigma(j')})$ and therefore  
	\begin{align*}
		v_{\sigma(j)}(A'_{\sigma(j)})\geq v_{\sigma(j)}(A'_{\sigma(j')}\setminus \set{g})
	\end{align*}
	making agent $\sigma(j)$ $\EF1$ towards $\sigma(j')$ also in this case.
\end{proof}

Since at the end of {Phase 2} the current allocation $\allocation$ is $\EF1$, so will be the output of Algorithm~\ref{algo:Case2} as {Phase 3} simply applies the envy-cycle elimination algorithm. 

The algorithms for Cases 1 and 2 imply Theorem~\ref{thm:EF1n2}.

\subsection{An $n$-approx.\ subject to weaker fairness criteria}
The main obstacle in obtaining a linear approximation is the allocation of the best good $o_i^1$, for each $i$, while ensuring fairness. In fact, this would provide a $O(n)$-approximation of $\delta_1$ and, hence, to $\opt$ in Case 1.  This can be circumvented by relaxing the envy condition on the agents' side requiring $\EF2$. In fact, assume we remove the goods $o_1^1, \cdots, o_n^1$, and run Algorithm \ref{algo:Case2} in the resulting instance. This provides a partial $\EF1$ allocation $\allocation$ that we are able to show to be $n$-approx.\ to $\delta_2$.
By simply assigning to each agent $i$ the corresponding $o_i^1$ we trivially get an $\EF2$ allocation. On the other hand, assigning each  $o_i^1$ to agent $i$ guarantees an $n$-approx. of $\delta_1$. These facts together imply:
\begin{restatable}{theorem}{EFtwoapprox}\label{thm:EF2approx}
	There exists a poly-time algorithm computing an $\EF2$ allocation providing a $n$-approximation to $\opt$.
\end{restatable}

The second relaxation we consider is epistemic $\EF1$ which implies $\PROP1$.

In the remainder of this section, we show the following.
\begin{restatable}{theorem}{approxEpistemic}\label{thm:approxEpistemic}
	There exists a poly-time algorithm computing an epistemic $\EF1$ and $\PROP1$ allocation providing a $n$-approximation to $\opt$. 
\end{restatable}

To this aim, we first provide a sufficient condition for epistemic $\EF1$. Recall that $\goods^h_i= \set{g_i^{(h-1)n+1}, \dots, g_i^{hn}}$ and $q$ is the largest integer such that $m\geq qn$.

\begin{restatable}{lemma}{sufficientEpistemic}\label{lemma:sufficientEpistemic}
	Let $\allocation$ be an allocation where each agent gets at least $q$ goods, and there exists $\set{x_1, \cdots, x_q}\subseteq A_i$ s.t.\ $x_h\in \goods_i^h$, for each $h\in[q]$. Then, $\allocation$ is epistemic $\EF1$.
\end{restatable} 
Thanks to $\Cref{lemma:sufficientEpistemic}$ we are ready to show \Cref{thm:approxEpistemic}.

\begin{proof}[Proof of \Cref{thm:approxEpistemic} -- Sketch] 
	Assume w.l.o.g\ $m=q\cdot n$, for some $q\in\mathbb{N}$; if not, we introduce an appropriate number of dummy goods evaluated $0$ and having a social impact of $0$.
	
	We express our problem as the one of finding a maximum weighted perfect matching on a bipartite graph. On one side, we have $m$ goods, on the other, we have $q$ copies of the agents. Specifically, for each agent $i$, we have the copies $i_1, \cdots, i_{q}$. The $h$-th copy of agent $i$ is connected to all goods in $\goods_i^h$,  with edge weights equaling the social impact of the agent for the corresponding good. 
	{ ${\mathbf\EF1}$.}
	By construction, any perfect matching corresponds to an allocation satisfying conditions of~\Cref{lemma:sufficientEpistemic}. 
	Therefore, a maximum weighted perfect matching in the bipartite graph guarantees epistemic $\EF1$ and hence $\PROP1$. 
	
	{ Approximation.}
	The constructed bipartite graph is $n$-regular. 
	As a direct consequence of Hall's Theorem,
	there exist $n$ disjoint perfect matchings that partition the set of edges of the graph. Therefore, the sum of the social values of these $n$ matchings (that equals $\sum_i s_i(\goods)$) is an upper bound to $\opt$.
	On the other hand, being the computed matching of maximum weight, it is also an $n$-approximation to the sum of the values of the $n$  matchings, and hence of $\opt$.
\end{proof}


\section{Socially Aware Agents}
So far we have shown that, being agents unconscious of their social imprint, it is not possible to guarantee a good social impact while pursuing fairness. 
Due to the growing attention to social issues, we might imagine scenarios where agents do take into account their and others' social impact when establishing their opinion on the outcome. 
To introduce some sort of social awareness, a first natural attempt might be to allow envy only towards agents having a worse social impact, that is, agents guaranteeing a better social impact are allowed to have a better bundle. Slightly more formally, we may say that an allocation $\allocation$ is envy-free (in a social aware sense) if for each agent $i$ one of the following conditions holds:
\begin{align*}
	v_i(A_i)\geq v_i(A_j)  && or && s_i(A_i) < s_j(A_j)\, .
\end{align*}
On the positive side, according to this definition, the envy an agent may have towards another agent is subjective (depends on her valuation), while who might be envied is objective (depends on the impact on society).
Therefore,
if an agent envies another this envy is eliminated because the envied agent has objectively a better impact on society. However, this first notion can be arbitrarily inefficient.
\begin{example}
	Assume there are two agents and one item. Both agents value $1$ the item; agent $1$ has a social impact $1$ while agent $2$ has social impact  $\epsilon$ for receiving the item. In this scenario, $\opt=1$. Consider the allocation $\allocation$ where the item is given to agent $2$. This provides a social welfare of $\epsilon$. Moreover, agent $1$ does envy agent $2$. However, the social impact of $1$ is higher than the one of agent $2$. If agents $1$ would decide to remove the envy towards agent $2$ on the sole basis of their current social impact, the allocation $\allocation$ would be considered fair while, for the sake of the social welfare, it is reasonable that agent $1$ envies agent $2$.
\end{example}

The main issue with this first attempt at defining a socially aware notion is that there is no comparison in what potentially could guarantee the possibly envious agent. Roughly speaking: It is foolish to believe we cannot envy someone just because she/he has a better impact than us if, in reality, we would make a better contribution if we were in her/his place. This is the idea behind the following definition.

\begin{definition}[Socially aware $\EF$]
	We say an allocation $\allocation$ is {\em socially aware envy-free} ($\SEF$) if for each pair of agents agents $i, j$ one of the following conditions holds:
	\begin{align*}
		v_i(A_i)\geq v_i(A_j)  && or && s_i(A_j) < s_j(A_j)\, .
	\end{align*}
	In turn, a socially aware agent $i$ envies $j$ ($i$ sa-envies $j$, in short) if $v_i(A_i)<v_i(A_j) $ and $s_i(A_j) \geq s_j(A_j)$.
\end{definition}

Notice that, if $s_i(g)=0$ for each $g\in\goods$ and $i\in\agents$ -- which means, agents have no social impact -- $\SEF \equiv \EF$; this implies $\SEF$ may not exist as well. We, therefore, relax this notion in the ``up to one good'' sense.

\begin{definition}[$\SEF$ up to one good]
	An allocation $\allocation$ is $\SEF$ {\em up to one good} ($\SEF1$) if for each pair of agents $i, j$, such that $A_j\neq \emptyset$, one of the following conditions holds:
	\begin{itemize}
		\item  there exists $g\in A_j$ such that $v_i(A_i)\geq v_i(A_j\setminus\set{g}) $, or 
		\item $s_i(A_j) < s_j(A_j)$.
	\end{itemize}
\end{definition}

We notice we could have defined the ``up to one good'' relaxation saying $i$ envies $j$ if, for each $g\in A_j$, $s_i(A_j\setminus \set{g}) \geq s_j(A_j)$ and $v_i(A_i)<v_i(A_j\setminus \set{g})$; however, this is weaker than the $\SEF1$ we just defined.

Surprisingly, we are able to show the following result.
\begin{restatable}{theorem}{sEFonemaxUt}\label{thm:sEF1maxUt}
	There exists a poly-time algorithm computing an $\SEF1$ allocation which is optimal for \maxUt.
\end{restatable}

\begin{proof}[Sketch]
	The theorem is a consequence of the following observation: 
	Given an optimal allocation $\OPT$, if $i$ sa-envies $j$, then $s_i(g) = s_j(g)$, for each $g\in \OPT_j$. 
	
	Our algorithm sequentially assigns the goods in $\goods$ optimally while maintaining $\EF1$. Assume at a certain step the current partial allocation $\allocation$ is $\EF1$ and $g$ has to be assigned to an agent $i^* \in \argmax_i s_i(g)$, to guarantee optimality. We need to make sure that the new allocation is $\EF1$. If not, there exists $j$ who would not be $\SEF1$ toward $i^*$, and hence not $\SEF$, in the $\allocation'$ obtained from $\allocation$ by allocating $g$ to $i^*$.
	$\allocation'$ is optimal, for the currently allocated goods, and our previous observation imply 1) $s_j(g)= s_{i^*}(g)$ and 2) $s_{j}(A_{i^*})= s_{i^*}(A_{i^*})$, as the social impact of $j$ is the same of $i$ for all goods in $A'_{i^*}= A_{i^*} \cup \set{g}$. 1) means that the only agents that might envy $i^*$ after allocating $g$ are the ones in $\argmax_i s_i(g)$, so, we should choose the $i^*$ who is not envied by any other agent in $\argmax_i s_i(g)$. However, it is possible that such an agent does not exist -- because there are no sources among $\argmax_i s_i(g)$ in the sa-envy-graph. 2) establishes that if there exists a cycle of sa-envy in $\argmax_i s_i(g)$ applying \cycleElimination\ does not affect the social impact. Hence, we can apply \cycleElimination\ until we find a $i^*\in \argmax_i s_i(g)$ who is not envied. This shows there always exists a way to sequentially assign goods optimally while maintaining $\EF1$.
\end{proof}


\section{Discussion}
We have investigated the fair division model when the underlying allocation problem has an impact on society.

Our work also paves the way for a large number of future challenges. One might study the scenario where goods are chores for the agents; in such case, not all our results naturally extend. In fact, the envy-cycle elimination no longer guarantees $\EF1$, and a different approach is needed~\cite{bhaskar2021approximate}. However, this cannot be immediately translated into an approximation algorithm for \maxUt.
Other appealing directions include the study of social welfare functions other than the utilitarian.

We have also introduced social awareness.  One may extend this definition to encode different levels of awareness by introducing a parameter $\alpha$ establishing how altruistic agents are. If $\alpha = 0$ we would have $\EF1$, which means, agents completely ignore their impact, while for $\alpha=1$, the notion would coincide with $\SEF1$. We have seen that an approximation to $\opt$ better than linear is not possible when requiring $\EF1$, while an outcome that is optimal and $\SEF1$ always exists. We wonder if there exists an explicit connection between the parameter $\alpha$ and the attainable approximation. 

\section{Acknowledgments}
This work was partially supported by:  
the PNRR MIUR project FAIR - Future AI Research (PE00000013), Spoke 9 - Green-aware AI; 
the PNRR MIUR project VITALITY (ECS00000041), Spoke 2 -  Advanced Space Technologies and Research Alliance (ASTRA);
the European Union - Next Generation EU under the Italian National Recovery and Resilience Plan (NRRP), Mission 4, Component 2, Investment 1.3, CUP J33C22002880001, partnership on ``Telecommunications of the Future''(PE00000001 - program ``RESTART''), project MoVeOver/SCHEDULE (``Smart interseCtions witH connEcteD and aUtonomous vehicLEs'', CUP J33C22002880001);
and
Gruppo Nazionale Calcolo Scientifico-Istituto Nazionale di Alta Matematica
(GNCS-INdAM).

\bibliographystyle{plain}
\bibliography{bibliography}

\newpage
\appendix
\section{Fair Division with Social Impact -- Supplemental Material }
In this supplemental material, we present the pseudo-code of the algorithms, further details about the guarantees of Algorithm~\ref{algo:EF1transfOrd}, and the proofs that have not been presented in the main body of the paper. Next, you will find a dedicated section for each of these.
\subsection{Missing Algorithms}
\begin{algorithm}[H]
	\SetNoFillComment
	\DontPrintSemicolon
	\KwIn{A fair division instance $\instance=(\agents, \goods, \set{v_i}_{i\in\agents})$}
	\KwOut{An $\EF1$ allocation $\allocation$}
	$\allocation\gets (\emptyset, \dots, \emptyset)$\\
	Build the envy graph $G$ corresponding to $\allocation$\\ \tcc{At this stage the set of edges is empty.}
	\While{$\goods\neq\emptyset$}{
		\While{There is no source in $G$}{
			$\allocation \gets \cycleElimination(\allocation, G)$ \\
			Build the envy graph $G$ corresponding to $\allocation$\\
		}
		$i \gets$ a source node in $G$, $g\gets$ any good in $\goods$\\
		$A_i\gets A_i \cup\set{g}$\\
		$\goods \gets \goods \setminus \set{g}$
	}
	\KwRet{$\allocation$}
	\caption{Envy-cycle elimination \label{algo:envyCycle}}
\end{algorithm}

\begin{algorithm}[H]
	\SetNoFillComment
	\DontPrintSemicolon
	\KwIn{$\OPT$, $\goods$, $\agents$, $\set{v_i}_{i\in \agents}$, $\set{s_i}_{i\in \agents}$}
	\KwOut{An allocation $\allocation$}
	$\allocation \gets (\emptyset, \cdots, \emptyset)$\\
	\While{$\goods\neq\emptyset$}{
		Build the envy graph $G$ corresponding to $\allocation$\\
		$\sigma \gets\topOrd(G)$\\
		$g\gets$ a good in $\goods$\\
		$i\gets \argmax_{i\in\agents}s_i(g)$,
		{ties are broken w.r.t.\ $\sigma$\label{line:tie-breaking}}
		$A_i\gets A_i \cup \set{g}$,
		$\goods \gets  \goods \setminus \set{g}$
	}
	\KwRet{$\allocation$}\\
	\caption{ \maxUt\ and $\SEF1$ allocation \label{algo:sEF1}}
\end{algorithm}


\subsection{Further Discussion about the General Approximation Algorithm subject to $\EF1$ for Ordered Valuations}
Algorithm~\ref{algo:EF1transfOrd} takes an initial allocation, denoted by $\allocation$, as input and outputs a new allocation, $\allocation'$, that satisfies the $\EF1$ property. The algorithm also guarantees to each agent $i$ a bundle $A_i'$ of social value at least $1/n$-th of the social value of her original bundle $A_i$. Formally, $n \cdot s_i(A_i')\geq s_i(A_i)$, for each $i\in\agents$. As a consequence, if the considered social welfare is a homogeneous function of degree $1$ in the utility of the agents,\footnote{That is, $\SW(\allocation) = f(u_1(A_1), \cdots, u_n(A_n))$ and $f(k\cdot x_1, \cdots, k\cdot x_n)= k \cdot f(x_1, \cdots, x_n) $.} then the output of the algorithm provides an $n$-approximation to the social welfare of the input allocation $\allocation$. Notably, egalitarian, utilitarian, and Nash social welfare functions all satisfy this homogeneity property.

For this reason, \Cref{prop:genApproxAlgIdentical} has as direct consequence the following.

\begin{theorem}
	\label{cor:orderedApprox}
	For ordered instances, given an $\alpha$-approximation for $\maxSW$, then, it is possible to compute an $\EF1$ $(\alpha\cdot n)$-approximation in poly-time,  for $\textsc{SW} \in \set{\ut, \eg, \nash}$. 
\end{theorem}

\Cref{cor:orderedApprox} has several consequences for ordered instances: i) Since the problem of maximizing the utilitarian is in $\classP$, it is sufficient to allocate each good to an agent with the highest social impact for it, an $n$-approximation to \maxUt\ is indeed possible subject to $\EF1$; ii) The problem of maximizing the Nash welfare is $\classNP$-hard, however, constant approximation algorithms are known, therefore, an $O(n)$-approximation to \maxNash\ subject to $\EF1$ is also possible; iii) The problem of maximizing the egalitarian welfare turns out to be more complicated, the best-known approximation is of $O(\sqrt{n}\log^3n)$, thus our approach leads to an $\EF1$ allocation guaranteeing an  $O(n^{\frac{3}{2}}\log^3n)$-approximation.

\subsection{Missing Proofs}
\identicalRB*
\begin{proof}
	We demonstrate that the allocation $\allocation$ can be generated by a sequential algorithm using an RB picking sequence $\sigma$. \Cref{fact:RBisEF1} establishes that RB picking sequences yield $\EF1$ allocations, this will prove the $\EF1$ property of allocation $\allocation$.
	
	Recall that the goods ordering is $g^1, \dots, g^m$.
	We construct the picking sequence $\sigma$ as follows: For every good $g^k$, we set $\sigma(k)$ to be the agent who receives $g^k$ in the allocation $\allocation$. Assuming ties are broken according to the goods ordering, the sequential algorithm that uses $\sigma$ as its picking sequence will output the allocation $\allocation$.
	
	Next, we need to show that $\sigma$ is an RB sequence, which is guaranteed by our block-based partitioning. In fact, let $k\in [m]$, and let $p_i$ and $p_j$ be the number of occurrences in the prefix $(\sigma(1), \cdots, \sigma(k))$ of agents $i$ and $j$, respectively. Let $h$ be $\lfloor \frac{k}{n}\rfloor$, in the first $k$ picks, every agent received exactly one good in each block $\goods^\ell$, for each $\ell\in [h]$, and possibly one more, if $n$ does not divide $ k$.
	Therefore,
	$p_i,p_j \in \set{\lceil \frac{k}{n}\rceil, \lfloor \frac{k}{n}\rfloor}$ implying $\modulus{p_i -p_j} \leq 1$. This shows that $\sigma$ is RB.
\end{proof}

\genApproxAlgIdentical*
\begin{proof}
	The $\EF1$ condition follows by \Cref{lemma:identicalRB}.
	
	Recall we assumed $m=q\cdot n$, otherwise, we add an appropriate number of dummy goods.
	The \textbf{for} loop on line 1 in Algorithm~\ref{algo:EF1transfOrd} is run $q$ times. At each run (round) the goods of $\goods^k$ are assigned one to each agent and agent $i$ receives the best good in $A_i\cap \goods^k$, if not empty. Let us denote by $\ell_i, \dots, \ell_h$ the rounds when $i$ received a good in $A_i$. Notice that in any other round $k$, $A_i\cap \goods^k=\emptyset$, otherwise, the algorithm assigns a good in $A_i$ to $i$. Hence, $s_i(A_i)= \sum_{j=1}^h s_i(A_i\cap \goods^{\ell_j})$, where $A_i\cap \goods^{\ell_j}$ is the subset of goods in $A_i$ that have been assigned in round $\ell_j$. In round $\ell_j$, $i$ gets the best good in $A_i\cap \goods^{\ell_j}$ and $\modulus{A_i\cap \goods^{\ell_j}}\leq \modulus{\goods^{\ell_j}} = n$. So, denoted by denoted by $o_{\ell_j}$ the good $i$ receives at round $\ell_j$,
	$s_i(A_i\cap \goods^{\ell_j}) \leq n\cdot s_i(o_{\ell_j})$, as $o_{\ell_j}$ has the best social impact among remaining goods of $A_i$.
	In conclusion, being $A_i'$ the bundle $i$ receives by the algorithm,
	\begin{align*}
	s_i(A_i)= \sum_{j=1}^h s_i(A_i\cap \goods^{\ell_j}) \leq \sum_{j=1}^h n\cdot s_i(o_{\ell_j}) \leq n\cdot s_i(A_i') \, , 
	\end{align*}
	and the thesis follows.
\end{proof}

\approxIdentical*
\begin{proof}
	We recall that, for identical valuations, an $\EFX$ allocation can be computed in polynomial time.
	Let $\allocation$ be such an allocation. Observe that any permutation of bundles in $\allocation$ to the agents is $\EFX$ because agents have the same valuations. We can therefore compute a maximum matching between agents and bundles of $\allocation$, where the weight of an agent-bundle edge is the social impact the agent would have by receiving that bundle. 
	Let us assume w.l.o.g.\ that agent $i$ is matched to bundle $A_i$ in the maximum matching. 
	
	Being $s_i(\goods)= \sum_{j=1}^n s_i(A_j)$, we have $ \sum_{i=1}^n \sum_{j=1}^n s_i(A_j) \geq \opt$.  In turn, since we obtained a maximum matching by assigning $A_i$ to agent $i$, any rotation of the bundles among the agents will provide a matching with a lower social impact.  
	Specifically, let $\allocation^k$, for $k=0, \dots, n-1$ an allocation such that $A_i^k= A_{\alpha(i,k)}$, where $\alpha(i,k)= i+ k$ if $i+k \leq n$ and  $\alpha(i,k)= i \mod n + k$, otherwise. Notice that $\allocation^0 = \allocation$. The maximality of the matching corresponding to $\allocation$ implies $\sum_{i=1}^n s_i(A_i)\geq \sum_{i=1}^n s_i(A^k_i)$, for each $k=0, \dots, n-1$.
	Summing up for all such $k$ we get 
	\begin{align*}
	n \cdot \sum_{i=1}^n s_i(A_i) \geq \sum_{k=0}^{n-1}  \sum_{i=1}^n s_i(A^k_i)
	&= \sum_{i=1}^n \sum_{j=1}^n s_i(A_j)\\
	&= \sum_{i=1}^n s_i(\goods) \geq \opt
	\end{align*}
	where the first equality hold true as in each $\allocation^0, \dots , \allocation^{n-1}$ agent $i$ gets a distinct bundle of $\allocation$.
	
	In conclusion, the allocation corresponding to the maximum matching guarantees an $n$-approximation to $\opt$.
\end{proof}

\Casetwosufficient*
\begin{proof}
	Let $\allocation$ be such an allocation.
	Recall that $\OPT_i$ is the disjoint union of $\set{B_i^k}_{k=1}^{k_i+1}$. Furthermore, $\delta_2=\sum_{i\in \agents}\sum_{k=2}^{k_i+1} s_i(B_i^k)$.
	
	As already observed, $s_i(o) \geq s_i(o')$ for each $o\in B_i^k$ and $o'\in B_i^{k+1}$, for $1\leq k \leq k_i$. Since in $A_i$ agent $i$ receives at least one good in $B_i^k$, for each $k=1, \cdots, k_i$, then $n\cdot s_i(A_i) \geq \sum_{k=2}^{k_i+1} s_i(B_i^k)$. Repeating this argument for each agent we get $n \cdot s(\allocation)\geq \delta_2$. Being $\delta_2 \geq \frac{\delta_1+ \delta_2}{2} \geq \frac{\opt}{2}$, the thesis follows. 
\end{proof}

\EFtwoapprox*
\begin{proof}
	Recall, given $\OPT$ an optimal solution of \maxUt, $\opt$ is the optimal value $\SW(\OPT)$. 
	Moreover, we have denoted by $m_i$ the number of goods in $\OPT_i$ and sorted the social value of such goods from the highest to the lowest, namely, $\opt^1_i\geq \opt^2_i\geq \cdots \geq \opt^{m_i}_i$. Hence, we denoted the good of value $\opt^k_i$ by $o^k_i$. 
	Furthermore, 
	$\Delta_1=\set{o^k_i \, \vert \, i\in\agents, k\in [n]}$ which has social value $\delta_1=\sum_{i=1}^n \sum_{k=1}^n \opt^k_i$ and  $\Delta_2= \goods \setminus \Delta_1$ has social value $\delta_2$.
	
	Assigning to each agent $i$ $o_i^1$ guarantees an $n$-approximation to $\delta_1$ as $n\cdot\opt^1_i \geq \sum_{k=1}^n \opt^k_i$.
	
	Recall, $B_i^k$ containing the goods $o_i^{(k-1)n +1}, \dots, o_i^{kn}$ and $\delta_2=\sum_{i\in \agents}\sum_{k=2}^{k_i+1} s_i(B_i^k)$. 
	
	Consider now the instance $\instance$ obtained by removing $o_1^1, \dots, o_n^1$ from the set of goods and let $\mathcal{O}$ be the allocation obtained from $\OPT$ removing the same goods. $\mathcal{O}$ is an optimal solution of \maxUt\ in the $\instance$. Let us run Algorithm~\ref{algo:Case2} on the instance $\instance$ giving  $\mathcal{O}$ as the optimal solution. The algorithm will partition the bundles of $\mathcal{O}$  into blocks of $n$ goods $C_i^k=\set{o_i^{(k-1)n +2}, \dots, o_i^{kn+1}}$ as long as it is possible. Notice that the goods in $C_i^k$ are at least as good as the ones in $B_i^{k+1}$, for each $k\geq 1$.
	Let $\allocation$ be the output of the algorithm.
	We know that the algorithm will assign each agent a bundle containing exactly one good in each $C_i^k$. Therefore, $n\cdot s_i(A_i) \geq \sum_{k=2}^{k_i+1} s_i(B_i^k)$ and the $n$-approximation to $\delta_2$ follows.
	
	Moreover, $\allocation$ is $\EF1$ for $\instance$ so it is a partial $\EF1$allocation for the original instance. Assigning $o_i^1$ to each agent $i$ will make the allocation $\EF2$ and guarantees an $n$-approximation to $\delta_1$. Putting these facts together we get an $n$ approximation as well.
\end{proof}

\sufficientEpistemic*
\begin{proof}
	Let us consider an agent $i$ and show that $A_i$ makes $i$ epistemic $\EF1$. Our proof proceeds in two steps. 
	
	Step 1: Let us first consider an ordered fair division instance with the same set of goods $\goods$ and the same set of agents $\agents$. In such ordered instance, the ranking of the goods is determined by the ranking of $i$, that is, $g_i^1 \succ g_i^2 \succ \dots \succ g_i^m$. Consider a run of an RB picking sequence $\sigma$ where $i$ always picks last, i.e.\ $\sigma(h\cdot n)=i$ for each $h\in [q]$. The resulting allocation $\allocation'$ is $\EF1$ in $i$'s perspective and agent $i$ during her $h$-th pick receives the good $y_h=g_i^{hn}$ which is the worst good in $\goods_i^h$. Moreover, $\modulus{A'_i}=q$.
	
	Step 2: We transform the allocation $\allocation'$ into $\allocation^*$ where the $\EF1$ condition is maintained for $i$ and $A^*_i=A_i$, concluding our proof as $\allocation^*$ would be the epistemic $\EF1$ certificate. The transformation sequentially proceeds as follows. Set $\allocation^*= \allocation'$. For each $h\in [q]$, let $j$ be the agent such that $x_h\in A'_j$, exchange $x_h$ with $y_h$ in $\allocation^*$. This exchange does not decrease the valuation $i$ has while maintaining $\EF1$.  In fact, $y_h$, in the perspective of $i$, is the worst good in $\goods_i^h$ and $x_h\in \goods_i^h$. 
	If there are still goods in $A_i\setminus A_i^*$, move those goods to $A_i^*$, this again does not decrease $i$'s utility, possibly decreasing her valuation for others' bundles. This maintains $\allocation^*$ $\EF1$ in $i$'s perspective. Being $A^*_i=A_i$, allocation $\allocation^*$ is the epistemic $\EF1$ certificate of the allocation $\allocation$ for agent $i$.
\end{proof}

\approxEpistemic*
\begin{proof}
	Assume w.l.o.g\ $m=q\cdot n$, if not, we introduce dummy goods evaluated $0$ and having a social impact of $0$ for all agents.
	
	To prove the theorem, we express our problem as the problem of finding a maximum weighted perfect matching on a bipartite graph. On one side, we have $m$ goods, on the other, we have $q$ copies of the agents. Specifically, for each agent $i$, we have the copies $i_1, \cdots, i_{q}$. The $h$-th copy of agent $i$ is connected to all goods in $\goods_i^h$  with the edge weights representing the social impact of the agents for the respective goods. 
	
	We first notice that the constructed bipartite graph is $n$-regular. Every copy $i_h$ is connected to the $n$ goods of $\goods_i^h$; moreover, for every good $g$, since only one copy of each agent is connected to $g$, the in-degree of $g$ is $n$. 
	
	Being our bipartite graph $n$-regular, we will make use of the following well-known fact, a direct consequence of Hall's Theorem:
	
	\begin{fact}\label{fact:dRegular}
		A $d$-regular bipartite graph always contains a perfect matching.
	\end{fact}
	
	Let us observe that any perfect matching in this bipartite graph corresponds to a complete allocation for the underlying fair division instance (but the vice-versa is not true), and the weight of the matching equals the social impact of the allocation. Thus, we can interchangeably refer to matchings and allocations in the rest of this proof.
	Moreover, by the construction of the graph, for any perfect matching in the defined bipartite graph, the corresponding allocation satisfies the conditions of~\Cref{lemma:sufficientEpistemic}. 
	
	Let $\mathcal{M}^*$ be a maximum weighted perfect matching in the bipartite graph. We next show that the corresponding allocation $\allocation^*$ ensures an $n$-approximation to \maxUt.
	
	To this aim, we will make use of the following observation. By \Cref{fact:dRegular}, there exist $n$ distinct perfect matching that partition the set of edges of the bipartite graph. To illustrate this, we employ a recursive approach: starting with the bipartite graph, which is $n$-regular, we extract a perfect matching, say $\mathcal{M}^1$. Subsequently, we remove the edges of $\mathcal{M}^1$ yielding a new bipartite graph that is $(n-1)$-regular. Repeating this process, we extract $n$ distinct matchings $\mathcal{M}^1, \cdots, \mathcal{M}^n$, that correspond to $n$ different allocations $\allocation^1, \cdots, \allocation^n$. Notice, that $\cup_{h=1}^n A_i^h= \goods$, for each $i\in \agents$, as each edge incident to a copy of $i$ belongs exactly to one of the $n$ matcghings. Hence, we derive the following inequalities:
	\begin{align*}
	n \cdot \SW(\allocation^*) \underset{(1)}{\geq}\sum_{h=1}^n \SW(\allocation^h) &=\sum_{i\in \agents} \sum_{h=1}^n s_i(A_i^h)  
	\\&\underset{(2)}{=} \sum_{i\in \agents} s_i(\goods) \geq  \opt \, ,
	\end{align*}
	where, (1) holds true since $\mathcal{M}^*$ is a maximum weighted perfect matching and (2) because, by the construction of the matchings $\mathcal{M}^1, \cdots, \mathcal{M}^n$, for each $g\in \goods$ and each agent $i$, there exists $h$ such that $g\in A_i^h$.
	In conclusion, $\allocation^*$ is an $n$-approximation to \maxUt\ and the thesis follows.
\end{proof}

\subsection{Maximum Utilitarian subject to $\SEF1$}
\sEFonemaxUt*
\Cref{thm:sEF1maxUt} is the result of the following discussion.

Let us start by understanding if (and how) socially aware envy arises in a socially optimal solution.  

\begin{observation}\label{prop:SEF1equalImpactItem}
	Let $\OPT$ be an optimal allocation for \maxUt. If $i$ sa-envies $j$ then $s_i(g) = s_j(g)$, for each $g\in \OPT_j$.
\end{observation}
\begin{proof}
	If $i$ sa-envies $j$ we have $s_i(\OPT_j) \geq s_j(\OPT_j)$.
	On the other hand, for each $g\in \OPT_j$, $s_i(g) \leq s_j(g)$ must hold because of optimality, otherwise moving $g$ to the bundle of $i$ would increase the social welfare. If there exists $g'\in A_j$ such that $s_i(g') < s_j(g')$ we get $s_i(\OPT_j) < s_j(\OPT_j)$ -- a contradiction.  
\end{proof}

\begin{corollary}\label{corollary:SEF1equalImpactBundle}
	Let $\OPT$ be an optimal allocation for \maxUt. If $i$ sa-envies $j$ then $s_i(\OPT_j) = s_j(\OPT_j)$.
\end{corollary}

We can conclude that when goods have distinct social impacts for different agents the social optimum is also $\SEF$. However, when this property does not hold, it is even possible that $\SEF$ might not exist. Hence we next focus on $\SEF1$ allocations.  

A first interesting consequence of \Cref{corollary:SEF1equalImpactBundle} is that if there exists a sa-envy-cycle in the socially-aware envy graph, eliminating envy the social welfare does not decrease -- the social welfare actually stays the same as the envious agent has a social impact for the envied bundle equal to the one of the owner. This constitutes the bulk of our algorithm to compute a \maxUt\ allocation that is also $\SEF1$.

Our algorithm proceeds as follows:
At each step, it computes the sa-envy graph $G$ and a topological order of $G$. It then allocates the next good to the agent having the highest social impact for that good, selecting the one coming first in the topological order in case of ties. Finally, the sa-envy graph $G$ of the resulting allocation is constructed and all the sa-envy cycles are iteratively deleted. A formal description of this algorithm is given by Algorithm~\ref{algo:sEF1}.

\begin{theorem}
	Algorithm~\ref{algo:sEF1} provides a $\SEF1$ and \maxUt\ allocation.
\end{theorem}
\begin{proof}
	Since at every iteration of the while loop the current good $g$ is allocated to an agent in $\argmax_{i\in\agents}s_i(g)$, \maxUt\ follows.
	
	Let us focus on $\SEF1$. At the first execution of the while loop, only one good is allocated, thus the current partial allocation $\allocation$ is $\EF1$ and hence $\SEF1$.
	
	Let us assume that at the end of the $k$-th execution of the while loop the allocation is $\SEF1$. We show that the $(k+1)$-th execution preserves the $\SEF1$ condition. Let $g$ be the good allocated during the $(k+1)$-th iteration and $\allocation$ the resulting partial allocation. Let $j \in \argmax_{i\in\agents}s_i(g)$ be the agent receiving $g$ -- hence $g\in A_j$. For the sake of a contradiction let us assume there exists $i$ sa-envying $j$ even after removing any good in $A_j$, and in particular $g$. 
	
	By \Cref{corollary:SEF1equalImpactBundle}, $s_i(A_j)=s_j(A_j)$. By \Cref{prop:SEF1equalImpactItem} $s_i(g)=s_j(g)$, implying, $i\in \argmax_{i\in\agents}s_i(g)$ and also $s_i(A_i)=s_j(A_j\setminus\set{g})$. This means that $i$ was sa-envious towards $j$ before allocating $g$ and hence the good $g$ should have been given to $i$ because of the tie-breaking rule on Line~\ref{line:tie-breaking}.
\end{proof}

\end{document}